\documentclass{article}
\usepackage{graphicx} 

\usepackage[english]{babel}
\usepackage[utf8]{inputenc}
\usepackage{amsmath,amsfonts,amssymb,amsthm}
\usepackage{tabto}
\usepackage[colorinlistoftodos]{todonotes}
\usepackage{blindtext}
\usepackage{subcaption}
\usepackage{caption}
\usepackage[nodisplayskipstretch]{setspace}
\usepackage{tikz-cd} 
\usepackage{tensor}
\usepackage{mathrsfs}
\usepackage{hyperref}
\usepackage[style=ext-authoryear-comp,maxcitenames=2,uniquename=false,backend=biber,uniquelist=false,articlein=false]{biblatex}
\usepackage[usenames,dvipsnames]{xcolor}
\usepackage{enumitem}
\usepackage{eucal}

\DeclareCiteCommand{\citeyear}
    {}
    {\bibhyperref{\printdate}}
    {\multicitedelim}
    {}

\hypersetup{
	colorlinks=true,         
	linkcolor=Black,          
	citecolor=MidnightBlue,
	urlcolor=MidnightBlue            
 }

\newtheorem{thm}{Theorem}

\makeatletter
\renewcommand\section{\@startsection{section}{1}{\z@}{-3.25ex plus -1ex minus -.2ex}{1.5ex plus .2ex}{\normalsize\bf}}
\renewcommand\subsection{\@startsection{subsection}{2}{\z@}{-3.25ex plus -1ex minus -.2ex}{1.5ex plus .2ex}{\normalsize\bf}}
\renewcommand\subsubsection{\@startsection{subsubsection}{3}{\z@}{-3.25ex plus -1ex minus -.2ex}{1.5ex plus .2ex}{\normalsize\bf}}
\makeatother

\title{Why Hadamard states?}
\author{Eleanor March\footnote{Department of Logic and Philosophy of Science, University of California, Irvine, CA. eleanor.march@uci.edu}}
\date{}

\begin{document}

\maketitle

\begin{abstract}
    In QFT on curved spacetime, the Hadamard condition is often presented as a necessary condition on `physically reasonable' states of the quantum field, and plays a central role in many theoretical and foundational applications---ranging from proofs of the renormalizability of Wick polynomials to derivations of the Hawking temperature. Yet despite this, the philosophical and foundational underpinnings of the Hadamard condition remain murky. I critically discuss existing motivations for the Hadamard condition in the literature, before arguing in favour of an alternative motivation for the Hadamard condition, according to which it is best understood as a necessary and sufficient condition for the existence of a well-defined operator product of a certain form on a sufficiently large space of observables of the quantum field, satisfying a variety of further conditions. This clarifies the role and status of the Hadamard condition, including its relationship to the equivalence principle, to well-definedness of physical quantities such as Wick polynomials and the expectation value of the stress--energy operator, and the sense in which Hadamard states are `vacuum-like'. 
\end{abstract}

\section{Introduction}\label{sec:intro}

In quantum field theory (QFT) on curved spacetime, locally covariant QFT, and semi-classical gravity, the Hadamard condition is often presented as a necessary condition on `physically reasonable', `vacuum-like' states of linear quantum fields. Hadamard states permit the definition of important physical quantities, such as Wick polynomials, time-ordered products, and the expectation value of the stress--energy operator on general curved spacetime backgrounds \parencite{Wald1977,Wald1978,Wald1994,BrunettiFredenhagen2000}. 
As such, many foundational results---ranging from derivations of Hawking radiation and the Hawking temperature \parencite{JanssenVerch2023,KayWald1991}, to quantum energy inequalities \parencite{Fewster2000,FewsterVerch2002,Fewster2017}, to proofs of the renormalizability of Wick and time-ordered products \parencite{HollandsWald2001,HollandsWald2002}---rely on the Hadamard condition.

Yet despite this, the philosophical and foundational underpinnings of the Hadamard condition remain murky. One `standard' motivation for the Hadamard condition found in the literature (see, e.g., \textcite{KayWald1991,Amir-HomayoonOttewill,FewsterVerch2013,KhavkineMoretti2015}) goes via appeal to a version of the strong equivalence principle---according to which the singularity structure of Hadamard states should `locally approximate' that of the Minkowski vacuum. But it is unclear what this kind of motivation has to do with, e.g., the well-definedness of Wick polynomials or the expectation value of the stress--energy operator. And it is doubly unclear what this version of the equivalence principle has to do with the classical equivalence principle, which states that the dynamics of matter fields in general relativity should `locally approximate' those in Minkowski spacetime (see, e.g., \textcite{March2025BPhil,Readetal2018,FletcherWeatherallpart2,MarchWeatherallFun1}). Another `standard' motivation for the Hadamard condition (see, e.g., \textcite{Fullingetal1978,KayWald1991,JanssenVerch2023,FewsterVerch2013}) appeals directly to the role of Hadamard states in defining quantities such as the expectation value of the stress--energy operator, Wick polynomials, time-ordered products, etc., but it is not clear how this motivation is supposed to underwrite the intended interpretation of the Hadamard condition as delimiting the class of `physically reasonable', `vacuum-like' states of linear quantum fields. 

In a companion paper \parencite{Marchtworoutes}, building on work by \textcite{Fewster2025Hadamard,Duetsch,FredenhagenRejzner2015,March2025BPhil,MarchWeatherallFun1}, I showed how both these motivations can be rendered technically precise, and substantially strengthened and generalized, from a mathematical physics standpoint. The aim of the present paper is (i) to critically assess these two motivations for the Hadamard condition as they are standardly presented in the literature, (ii) to explain what is added to them by the results of \parencite{Marchtworoutes}, and (iii) to draw out the interpretational consequences of these results for the role and status of the Hadamard condition in QFT on curved spacetime. Ultimately, I argue that the 
Hadamard condition is best understood as a necessary and sufficient condition for the existence of a well-defined operator product of a certain form on a `sufficiently rich' space of observables of the quantum field, satisfying a particular collection of other assumptions. (From this, one can recover the arguments about the role of Hadamard states in defining e.g., the expectation value of the stress--energy operator, Wick polynomials, etc., as a special case.) Most of these assumptions, I argue, are conditions which we should expect to hold of any `physically reasonable' linear QFT; the remainder clarify the sense in which Hadamard states can be thought of as `vacuum-like'.

As such, the structure of this paper will be as follows. In \S\ref{sec:prelims}, I provide some relevant technical background on the construction of QFTs from Green-hyperbolic field theories in the perturbative algebraic setting, including the general definition of the Hadamard condition in this setting due to \textcite{Fewster2025Hadamard}. I then, in \S\ref{sec:why} discuss in detail the two `standard' motivations for the Hadamard condition in the literature mentioned above, and raise some worries about both strategies. This takes us to \S\ref{sec:operatorstohadamard}, where I set out and discuss my preferred answer to the question `why Hadamard states?', and show how it avoids the worries raised in \S\ref{sec:why}. I close in \S\ref{sec:conclusion}.

\section{Background}\label{sec:prelims}

\subsection{Preliminaries}

We work with a finite rank complex vector bundle $V\rightarrow E\overset{\pi}{\rightarrow} M$ over a globally-hyperbolic, time-oriented, and oriented Lorentzian $n$-manifold $(M,g)$, where $M$ is assumed paracompact and with finitely many connected components. $E^*$ denotes the dual bundle, and $\mathrm{dvol}_g$ the density on $M$ induced by $g$. $E$ is assumed to be \emph{Hermitian}, which means that it is equipped with a smooth non-degenerate bundle metric $(\,\cdot\,,\,\cdot\,)$, and a base-preserving antiunitary involutive bundle automorphism $C:E\rightarrow E$ (complex conjugation) such that $(C(v),C(w))_x=\overline{(v,w)}{}_x$. One can then define a bilinear bundle metric $\langle\,\cdot\,,\,\cdot\,\rangle$ via $\langle v,w\rangle_x:=(C(v),w)_x$. Both $(\,\cdot\,,\,\cdot\,)$ and $\langle\,\cdot\,,\,\cdot\,\rangle$ can be extended to pairings between sections of $E$ with compactly intersecting supports via integration against $\mathrm{dvol}_g$.

For any $S\subset M$, $J^\pm(S)$ denotes its causal future/past. $J^\pm M$ denotes the bundle of (future, past) causal covectors on $M$; $N^\pm M$ denotes the bundle of (future, past) null covectors on $M$; vector bundles overset with a dot denote removal of the zero section. For any finite rank vector bundle $B$ over $(M,g)$, $\Gamma(B)$ denotes its space of smooth sections, $\Gamma_\mathrm{c}(B)$ its space of smooth sections with compact support, $\Gamma_{\mathrm{sc},\pm}(B)$ its space of smooth sections with support in $J^\pm(S)$ for some compact $S\subset M$, and $\Gamma_\mathrm{sc}(B)$ its space of smooth sections with support in $J(S)$ for some compact $S\subset M$. (These are all Fr\'echet topological vector spaces \parencite{Baer2015}.) 

Given any two Hermitian vector bundles $B_1, B_2$ and a linear map $S:\Gamma(B_1)\rightarrow\Gamma(B_2)$, the \emph{formal Hermitian adjoint} of $S$, if it exists, is the unique linear map ${}^\dag\! S:\Gamma(B_2)\rightarrow\Gamma(B_1)$ such that $(f_1,S(f_2))_2=({}^\dag\! S(f_1),f_2)_1$ for any $f_1\in \Gamma(B_1)$, $f_2\in \Gamma(B_2)$ such that $S(f_1)$ and $f_2$ have compactly intersecting supports. The \emph{formal transpose} ${}^t\! S$ of $S$, if it exists, is defined via $\langle f_1,S(f_2)\rangle_2=\langle{}^t\! S(f_1),f_2\rangle_1$. The formal Hermitian adjoint and formal transpose generalize the usual notions of the Hermitian adjoint and transpose of a linear map between vector spaces. We also write $\overline{S}=C_2SC_1$, i.e., the complex conjugate of $S$.

\subsection{Green-hyperbolic operators}

We consider the general class of linear differential operators $P:\Gamma(E)\rightarrow\Gamma(E)$. For Hermitian vector bundles over oriented Lorentzian manifolds, this encompasses all linear differential operators which arise as the Euler--Lagrange operator of some Lagrangian. Any such operator has a formal Hermitian adjoint. If ${}^\dag\! P=P$, $P$ is called \emph{formally Hermitian}. Now suppose that $P$ is formally Hermitian and that there exist linear maps $G^\pm_P:\Gamma_\mathrm{c}(E)\rightarrow\Gamma_{\mathrm{sc},\pm}(E)$ such that for any $f\in\Gamma_{\mathrm{c}}(E)$, (i) $PG^\pm_P(f)=f=G^\pm_PP(f)$, and (ii) $\mathrm{supp}(G^\pm_P(f))\subset J^\pm(\mathrm{supp}(f))$. Then $P$ is a \emph{formally-Hermitian Green-hyperbolic operator} \parencite{Fewster2025Hadamard}, and $G^\mp_P$ are called the \emph{advanced} (respectively \emph{retarded}) \emph{Green functions} for $P$. The advanced and retarded Green functions of a Green-hyperbolic operator are unique \parencite{Khavkine2014,Baer2015}; they are also not really functions but distributional (generalized) sections---a point which we will return to in the next section. In what follows, we will also assume that $P$ is real, i.e., $P=\overline{P}:=CPC$. In this case, one has that ${}^tG^\pm_P=G^\mp_P$, i.e., the advanced and retarded Green functions turn into one another under exchange of their arguments.

Roughly, one can think of a Green-hyperbolic differential operator as one whose inhomogeneous problem (i.e., the problem of solving $P(\varphi)=f$ for some `source' $f$) is well-posed, in the following sense:\footnote{The argument that \emph{all} hyperbolic linear differential operators which admit a well-posed initial value problem are Green hyperbolic is known as \emph{Duhamel's principle}, see \textcite{Khavkine2014}.} given any compactly supported `source' section $f$ of $E$, $\varphi=E^\pm_Pf$ is the unique solution of $P(\varphi)=f$ with past (respectively future) compact support (meaning that the intersection of $\mathrm{supp}(\varphi)$ with $J^\mp(x)$ for any $x\in M$ is compact). For our purposes, however, the most important property of real formally-Hermitian Green-hyperbolic operators is that it is for this class of operators that the causal propagator can be defined. (The causal propagator should be familiar from textbook treatments of QFT, though the set-up here is rather more abstract than what is usual.) Let $P$ be a real formally-Hermitian Green-hyperbolic operator. We define the \emph{causal} (advanced-minus-retarded) Green function for $P$ as the difference $G_P:=G^-_P-G^+_P:\Gamma_\mathrm{c}(E)\rightarrow\Gamma_{\mathrm{sc}}(E)$. The causal Green function has several useful properties. First, it is antisymmetric, i.e., ${}^tG_P=-G_P$. Second, for any $f\in \Gamma_\mathrm{c}(E)$, $P G_P(f)=\mathbf{0}=G_PP(f)$; or in words, the causal Green function encodes both the space of solutions of $P$ (as its image), and the space of sources for the inhomogeneous problem for $P$ (as its kernel). From the causal Green function, one can define the \emph{causal propagator} for $P$, which is an antisymmetric bilinear form $\Delta_P(\,\cdot\,,\,\cdot\,)$ defined via
\begin{equation*}
    \Delta_P(f,g):=\langle f, G_P(g)\rangle,
\end{equation*}
for all $f,g\in\Gamma_\mathrm{c}(E)$. The causal propagator can be regarded as a distributional section of the external tensor product bundle $E^*\boxtimes E^*$, i.e., the bundle $V^*\otimes V^*\rightarrow E^*\boxtimes E^*\rightarrow M\times M$ defined by taking the tensor product of the pullbacks of $E$ along the projections of $M\times M$ onto each factor. One can think of  $E^*\boxtimes E^*$ as a formal way of encoding, roughly, the space of $E^*\otimes E^*$-valued functions in two variables on $M$. 

\subsection{Hadamard two-point functions}\label{sec:Hadamardtwopoint}

To introduce the Hadamard condition in this general context, we now need to discuss the distributional character of the causal propagator in some more detail, which will require some tools from microlocal analysis. (To keep the presentation here as accessible as possible, I will introduce these tools in broad-brush terms; detailed definitions can be found in e.g., \textcite{Fewster2025Hadamard,Hoermander1,Radzikowski,FredenhagenRejzner2015,Brouder2014}.) Recall that a \emph{distributional section} of a complex finite rank vector bundle $B$ over $M$ is a continuous linear map $u:\Gamma_c(B^*\otimes\wedge^nT^*M)\rightarrow\mathbb{C}$, where $\wedge^nT^*M$ is the bundle of volume forms on $M$. Every smooth section of $B$ is a distributional section of $B$, since any smooth section of $B$ can be integrated against a compactly supported section of $B^*$ and a volume form to yield a complex number. Such distributional sections are called \emph{regular}. The remainder are called \emph{singular}. 

Characterizing where and how a distributional section fails to be regular is the aim of microlocal analysis. An important tool for doing this is the \emph{wavefront set} of a distribution \parencite{Hoermander1}. Any distributional section $u$ of $B$ can be associated, uniquely, with a certain (closed, conic) subset $\mathrm{WF}(u)\subseteq\dot{T}^*M$, called the wavefront set of $u$. Roughly, one can think of $\mathrm{WF}(u)$ as encoding both the locations (on $M$) where $u$ is singular, and the directions in which $u$ is singular (as measured by the failure of $u$ to be localized in those directions in Fourier space). For our purposes, the wavefront set allows us to do two things. First, it allows us to articulate conditions under which the product of two distributions is well-defined as a distribution (which is not the case in general). These conditions are known as \emph{H\"ormander's criterion}, and will be relevant for our discussion of operator products in \S\ref{sec:observables} and in \S\ref{sec:operatorstohadamard}. Second, following \textcite{Fewster2025Hadamard} (in turn building on the work of \textcite{Radzikowski}), the wavefront set allows us to define the Hadamard condition for a general class of real formally-Hermitian Green-hyperbolic operators going well beyond the normally-hyperbolic ones such as Klein--Gordon theory for which the condition was first introduced.

To begin, we say that a real formally-Hermitian Green-hyperbolic operator $P$ on a Hermitian vector bundle $E$ is $\mathcal{V}^+ $\emph{-decomposable} if there exists a conic, relatively closed subset $\mathcal{V}^+$ of $\dot{T}^*M$ (so that $\mathcal{V}^+\cap \mathcal{V}^-=\varnothing$, where $\mathcal{V}^-:=-\mathcal{V}^+$) such that 
\begin{equation*}
    \mathrm{WF}(\Delta_P)\subset(\mathcal{V}^+\times \mathcal{V}^-)\cup(\mathcal{V}^-\times \mathcal{V}^+).
\end{equation*}
Roughly, $\mathcal{V}^+$-decomposability says that the singular structure of $\Delta_P$ can be split into a well-defined `future-past' and `past-future' component, but where these are not required to be defined with respect to the metric lightcones. As such, any $\mathcal{V}^+$-decomposable real formally-Hermitian Green-hyperbolic operator can associated with an `intrinsic' causal (i.e., lightcone-like) structure; namely, we say that a covector $k$ at $x\in M$ is $\mathcal{V}$\emph{-causal} if it lies in the (fibrewise) convex hull $\mathcal{U}^\pm$ of $\mathcal{V}^\pm$ and $\mathcal{V}$\emph{-spacelike} otherwise.

We can now introduce the Hadamard condition. Let $E$ be a Hermitian vector bundle and $P$ a real formally-Hermitian Green-hyperbolic operator thereon. We say that a distributional section $\Delta_H$ of $E^*\boxtimes E^*$ is a (bosonic) $\mathcal{V}^+$\emph{-Hadamard} two-point function for $P$ iff
\begin{itemize}
    \item (Antisymmetric part) $\Delta_H-{}^t\Delta_H=\Delta_H-\overline{\Delta}_H=2i\Delta_P$;
    \item (Wavefront set spectrum condition (WFSSC)) Its wavefront set is a subset of one causal half of that of the causal propagator for $P$, i.e.,
    \begin{equation*}
        \mathrm{WF}(\Delta_H)\subset\mathcal{V}^+\times\mathcal{V}^-;
    \end{equation*}
    \item (Bi-solution) It is a distributional bi-solution of $P$, in the sense that for any $f,g\in\Gamma_c(E)$, $\Delta_H(P(f),g)=0=\Delta_H(f,P(g))$; and
    \item (Positivity) It is positive semi-definite, i.e.\ for any $f\in \Gamma_\mathrm{c}(E)$, $\Delta_H(\overline{f},f)\geq 0$.
\end{itemize}
Roughly, a (bosonic) $\mathcal{V}^+$-Hadamard two-point function can be thought of as the kind of object which could be the two-point function of a state (by (positivity))---though note that, at this point, we have not specified an algebra that it could be a state \emph{on}); has as its singular structure one causal half of that of the causal propagator for $P$ (WFSSC); and shares some other features in common with the causal propagator for $P$ ((antisymmetric part), up to a factor of $\frac{i}{2}$, and (bisolution)). I will discuss the relationship between a (bosonic) Hadamard two-point function and a Hadamard state further in \S\ref{sec:observables}. 

\subsection{Observables}\label{sec:observables}

Finally, I will discuss how a real formally-Hermitian Green-hyperbolic field theory can be used to construct a bosonic quantum field theory (for the extension to the fermionic case, see \textcite[Appendix A]{Marchtworoutes}). This requires us, in the first instance, to define an appropriate space of observables of the theory. Given a Hermitian vector bundle $E$, an \emph{observable} for $E$ is a continuous map $A:\Gamma(E)\rightarrow\mathbb{C}$. That is, an observable is an assignment of a complex number to every possible field configuration. This definition of an observable can be specialized in various ways. In what follows, we will focus on observables which are \emph{polynomial}, in the sense that there exist compactly-supported distributional sections $\alpha^{(i)}$ of $\boxtimes^{i}_\mathrm{sym}E$, $i=0,1,2,...$, such that $A(f)=\sum_{i=0}\langle\alpha^{(i)},\otimes^{i}f\rangle$ for all $f\in\Gamma(E)$, where the subscript $\mathrm{sym}$ denotes symmetrization. In other words, a polynomial observable is one which is a (formal) power series in sections of $E$. The space of polynomial observables on $E$ is denoted $\mathcal{P}(E)$, or, for polynomial observables which are (formal) power series in some parameter, $\hbar$, $\mathcal{P}(E)[[\hbar]]$.

Two subspaces of $\mathcal{P}(E)$ will be important in this paper. First, the \emph{regular polynomial observables} $\mathcal{P}_\mathrm{reg}(E)$, for which every coefficient $\alpha^{(i)}$ is a regular distributions. Second, for any conic, relatively closed $\mathcal{V}^+\subset \dot{T}^*M$ the $\mathcal{V}$\emph{-microcausal polynomial observables} $\mathcal{P}_{\mathcal{V}\mathrm{mc}}(E)$,\footnote{This definition generalizes the definition of a `microcausal observable' in \parencite{FredenhagenRejzner2015,Rejzner2022}, as well as the definition of a `field' in \parencite{Duetsch}.} for which the wavefront set of each $\alpha^{(i)}$ is disjoint from those points where all $i$ wavevectors are in the closed future cone or closed past cone of $\mathcal{V}$-causal covectors, i.e.,
\begin{equation*}
    \mathrm{WF}(\alpha^{(i)})\cap (\times^i\mathcal{U}^+\cup \times^i\mathcal{U}^-)=\varnothing,
\end{equation*}
where as before $\mathcal{U}^\pm$ is the (fibrewise) convex hull of $\mathcal{V}^\pm$. (To forestall any possible misconceptions going forward, it is worth emphasizing that, despite the terminology, the $\mathcal{V}$-microcausal polynomial observables have nothing to do with the condition which goes by the name `microcausality' found in axiomatic approaches to QFT on Minkowski spacetime.) An example of a $\mathcal{V}$-microcausal polynomial observable which is not a regular polynomial observable is the squared field amplitude, for which $\alpha^{(2)}(x,y)=\delta(x,y)\langle\,\cdot\,,\,\cdot\,\rangle$ and $\alpha^{(i)}=\mathbf{0}$ otherwise. Given any real formally-Hermitian Green-hyperbolic operator on $E$, one can also restrict the (regular, $\mathcal{V}$-microcausal) polynomial observables to the space of solutions of $P$ (identifying observables which agree on that space of solutions to avoid degeneracies)---the resulting spaces will be called the \emph{on-shell} (regular, $\mathcal{V}$-microcausal) polynomial observables for $P$, denoted $\mathcal{P}(E,P)$ (respectively $\mathcal{P}_\mathrm{reg}(E,P)$, $\mathcal{P}_{\mathcal{V}\mathrm{mc}}(E,P)$).

We now discuss algebraic structures on these spaces of observables.
First, note that classically, any $\mathcal{V}$-decomposable real formally-Hermitian Green-hyperbolic operator $P$ on $E$ is equipped with the \emph{Poisson--Peierls bracket} $[\,\cdot\,,\,\cdot\,]_P$ on $\mathcal{P}_{\mathcal{V}\mathrm{mc}}(E)$ via
\begin{equation}
    [A_1,A_2]_P:=\mathrm{prod}\circ \left\langle \Delta_P, \delta\otimes\delta\right\rangle A_1\otimes A_2,\label{eq:Peierls}
\end{equation}
where $\mathrm{prod}$ is the multiplication operator, i.e., $\mathrm{prod}\circ(A_1\otimes A_2)(f)=A_1(f)A_2(f)$, and $\delta$ denotes the functional derivative. A full discussion of the sense in which the Poisson--Peierls bracket is a generalization of the more familiar Poisson bracket would take us a little too far afield here; for now, the point is that the Poisson--Peierls bracket permits the definition of a Poisson bracket-like structure on the space of $\mathcal{V}$-microcausal polynomial observables without the introduction of canonical variables \parencite{Peierls1952}. Note that $[\,\cdot\,,\,\cdot\,]_P$ defines a symplectic form on $\mathcal{P}_{\mathcal{V}\mathrm{mc}}(E,P)$ \parencite{Fewster2025Hadamard,Duetsch,Khavkine2014}. 

Any distributional section $\Pi$ of $E^*\boxtimes E^*$ induces a unital associative algebra product on $\mathcal{P}_\mathrm{reg}(E)[[\hbar]]$, denoted $\star_\Pi$, via
\begin{equation}
    A_1\star_\Pi A_2:=\mathrm{prod}\circ\mathrm{exp} \left(\hbar \left\langle \Pi, \delta\otimes\delta\right\rangle\right)A_1\otimes A_2\label{eq:starprod}
\end{equation}
where $\mathrm{exp}$ is in the sense of formal power series. This will be relevant in \S\ref{sec:operatorstohadamard}, when we come to discuss conditions that this operator product should satisfy in quantization; for now, one can use this definition to make the connection between the definition of a (bosonic) Hadamard two-point function in \S\ref{sec:Hadamardtwopoint} and the notion of a Hadamard \emph{state}. Suppose that $\Pi$ satisfies (antisymmetric part) and (bisolution). Then any bosonic $\mathcal{V}^+$-Hadamard two-point function $\Delta_H$ defines, uniquely, a (quasifree) Hadamard \emph{state} on the algebra $(\mathcal{P}_\mathrm{reg}(E,P)[[\hbar]], \star_{\frac{i}{2}\Delta_P})$,\footnote{Recall that for a quasifree state, one has $\omega_{2n+1}=\mathbf{0}$ for all $n\in\mathbb{N}$ and $\omega_{2n}(f_1,...,f_{2n})=\sum_\mathrm{partitions} \omega_2(f_{i_1},f_{i_2})...\omega_2(f_{i_{2n-1}},f_{i_{2n}})$ for all $n>0\in\mathbb{N}$, where $\omega_n$ is the $n$-point function i.e., a quasifree state is one for which Wick's theorem holds. The definition of the $n$-point functions is essentially standard but would take us a little too far afield here; the reader should refer to e.g., \textcite{KhavkineMoretti2015}.} with $\Delta_H$ as the two-point function of that state \parencite{Fewster2025Hadamard,Duetsch,FredenhagenRejzner2015}, where we recall that a state on a unital star algebra $(\mathcal{A}, \star)$ is a linear map $\omega:\mathcal{A}\rightarrow \mathbb{C}$ taking the multiplicative unit to $1$ and such that $\omega(\overline{A}\star A)\geq 0$ for all $A\in\mathcal{A}$. If $\Pi=\Delta_H$, then the definition of this state takes a particularly simple form: it is the map $\langle\,\cdot\,\rangle_\mathbf{0}:\mathcal{P}_{\mathcal{V}mc}(E)[[\hbar]]\rightarrow \mathbb{C}[[\hbar]]$, $\langle A\rangle_\mathbf{0}:=A(\mathbf{0})$, which evaluates every observable at the zero section. For the case of Klein--Gordon theory, this leads to an alternative (and equivalent) expression of the Hadamard condition in terms of the `Hadamard form' of the two-point function \parencite[theorem 5.1]{Radzikowski}, which was used in much of the older literature (see, e.g., \textcite{KayWald1991,Fullingetal1978,Wald1994}). 

\section{Why Hadamard states?}\label{sec:why}

With this technical background under our belt, we can now return to the question posed at the beginning of this paper: why Hadamard states? I will consider two popular motivations given for the Hadamard condition---one which appeals to a version of the equivalence principle, and the other which is more pragmatic in character (appealing to the use of Hadamard states in various theoretical applications). These are not the only arguments for the Hadamard condition which have appeared in the literature (though they are the primary ones); another argument which I will not consider here is due to \textcite{FewsterVerch2013}.\footnote{Though some of the worries I will discuss in this section also apply to this motivation---in particular, that it is not clear what it has to do with the sense in which Hadamard states are supposed to be `vacuum-like'.} I will argue that both motivations are unsatisfactory---either because the way they are supposed to work is unclear, or because it is not clear how they support the intended physical and foundational significance of the Hadamard condition as delimiting the class of `physically reasonable', `vacuum-like' states. 

As a brief health warning before I begin: it is important to distinguish at the outset two questions which fall under the umbrella `why Hadamard states?'. First, one might ask: what heuristics, principles, and arguments led us to delimit the class of Hadamard states as the `physically reasonable' ones in the first place? And second, one might ask: what is the argument for imposing the Hadamard condition (and thinking that only Hadamard states are `physically reasonable) \textit{after the fact}---that is, after having formulated (via whatever route) the Hadamard condition? Both questions are interesting, but my focus throughout this paper will be on the latter (in more standard philosophy of science parlance, this is a `context of justification' project rather than a `context of discovery' one). In particular, my criticisms of the various motivations for the Hadamard condition considered in this section should not be taken as criticisms of these motivations construed as answers to the first question.\footnote{It is worth noting, in the context of the motivation via the equivalence principle I will consider, that even critics of the equivalence principle agree that it is unproblematic as a heuristic \parencite{FletcherWeatherallpart2,Fletcher2020,Weatheralldogmas}.}

As a second health warning: one might worry that my distinguishing the equivalence principle and pragmatic motivations here is somewhat artificial, given that there is a large amount of overlap in the physicists endorsing both. Nevertheless, I think that there is some value in considering each motivation separately, since the problems facing each are rather different in character, and it will allow us to see more clearly which of these problems could be resolved be taking the two motivations together, and which could not. 

\subsection{The equivalence principle}\label{sec:equivalence}

The first motivation for the Hadamard condition I will consider appeals to a version of the equivalence principle. On this approach, the argument for the Hadamard condition goes roughly as follows: since the two-point function of the Minkowski vacuum state for the Klein--Gordon field satisfies WFSSC for $\mathcal{V}^+=\dot{N}^+$, `physically reasonable' vacuum-like states of the Klein--Gordon field on an arbitrary curved spacetime should also satisfy WFSSC (for this choice of $\mathcal{V}^+$), since, by the equivalence principle, their singularity structure should `locally approximate' that of the Minkowski vacuum. This idea was heuristically important in the development of the Hadamard condition, but has also been taken up as an argument for imposing the Hadamard condition after the fact, e.g.:
\begin{quote}
    The momentum-space approach arose from the philosophy that the ultraviolet divergences of QFT arise from short-wavelength modes which do not probe the curvature of space-time. Thus, possession of the Hadamard form can in some suitably heuristic sense be described as an expression of the equivalence principle for quantum fields in a curved space-time background. \parencite[1943]{Amir-HomayoonOttewill}
\end{quote}
\begin{quote}
    We suggest that the Hadamard condition is a necessary (and, for quasifree states, necessary and sufficient) condition for a state to be physically admissible. This condition [...]\ is motivated in part by the equivalence principle [...]. \parencite[70]{KayWald1991}
\end{quote}
And in the more recent literature:
\begin{quote}
    One of the fundamental difficulties in the theory of quantum fields on curved spacetimes is that generic spacetimes possess no symmetries that could serve to distinguish a preferred vacuum state. Instead, for linear fields, experience has led to the delineation of the class of Hadamard states [...] whose short-distance structure approximates that of states with finite energy density in Minkowski space, motivated by the equivalence principle. \parencite[1]{FewsterVerch2013}
\end{quote}
\begin{quote}
    The idea is that, although it is not possible to uniquely assign each spacetime with a physically distinguishable state, it is possible to select a type of divergence in common with all physically relevant states is every spacetime. These preferred quasifree states with the same type of divergence ``resembling" Minkowski vacuum in a generic spacetime are called Hadamard states. \parencite[227]{KhavkineMoretti2015}
\end{quote}

To begin, I would like to raise two worries about this kind of argument. The first has to do with the version of the equivalence principle which is being invoked. In the context of classical field theory, the strong equivalence principle (SEP) is usually understood as asserting a certain relationship between, on the one hand, the dynamics of a general-relativistic matter theory set on an arbitrary relativistic spacetime $(M,g)$, and on the other hand, those same dynamics set on Minkowski spacetime \parencite{Brown2005,FletcherWeatherallpart2,Readetal2018,March2025BPhil,GhinsBudden2001,MarchWeatherallFun1}. The theory is said to satisfy SEP just in case those dynamics agree with one another at a point whenever the two metrics agree, to first order, under the action of some (local) diffeomorphism there (and hence `locally approximate' one another in sufficiently small neighbourhoods of that point---see \textcite[\S\S 5.1, 5.3]{March2025BPhil} for details and further discussion of this way of explicating SEP). In other words, in the context of classical field theory, SEP has something to do with the \textit{dynamics} of matter fields, whereas the version of the equivalence principle at issue in the claim that the singularity structure of Hadamard states `locally approximates' that of the Minkowski vacuum does not.\footnote{One might suggest that perhaps this means that we need a more expansive understanding of SEP, which would apply to structures other than dynamics, e.g., the singularity structure of a distribution or suchlike. I am not unsympathetic to this view, and, as will become clear in what follows, my criticisms of the equivalence principle argument for the Hadamard condition will not ultimately turn on it. The point here is just that, in the absence of further explication of SEP, it is not clear from the way that it is applied in the context of classical field theory that it has anything to do with the claim that the singularity structure of Hadamard states `locally approximates' that of the Minkowski vacuum.} 

The second worry has to do with the relationship between this motivation and the ways in which the Hadamard condition is used in the literature. At a first pass, the problem here is that it is not clear what the equivalence principle (or, more precisely, the version of it which is supposed to apply to the singularity structure of the Minkowski vacuum) has to do with, e.g., well-definedness of the expectation value of the stress--energy operator, Wick polynomials, etc., which are the main reasons the Hadamard condition is imposed.\footnote{One might, in my view correctly, suggest that this is not even a worry until we have a defensible version of SEP to hand which does apply to the singularity structure of the Minkowski vacuum. Nevertheless, I think there is some value in getting this worry on the table now, since it can be considerably sharpened in light of the results of \parencite{Marchtworoutes} (to be discussed presently). I am grateful to an anonymous referee for pressing me on this point.} Now, to be clear, I am not claiming that one ought to demand that this relationship be somehow `obvious' or \textit{a priori} (whatever that would mean). It is, of course, entirely possible for foundational principles and conditions to be interrelated in non-obvious ways. Rather, the point is that one would like some clarity on what, exactly, the equivalence principle is doing for us in ensuring well-definedness of the expectation value of the stress--energy operator, Wick polynomials, etc., which is not simply to restate the fact that a distribution `having the right singular structure' is related to the well-definedness of products involving that distribution. The latter is unobjectionable, but it does not shed light on the role of the equivalence principle (since the criterion of `having the right singular structure' could just be applied at each spacetime directly).

Now in fact, as I show in \parencite{Marchtworoutes}, the first of these worries can be overcome. In particular, drawing on a statement of SEP due to \textcite{March2025BPhil,MarchWeatherallFun1} (see also \textcite{MarchWeatherallcovgauge,MarchWeatherallnaturaltheories} for background), originally developed in the context of matter dynamics on relativistic spacetimes, I show that this statement can be straightforwardly extended to apply to the wavefront set of a two-point function. Using this, I prove the following result:
\begin{thm}[March 2026]\label{thm:SEPforhadamard}
    Let $\mathscr{E}:\mathsf{Loc}\rightarrow \mathsf{FB}$ be a natural vector bundle, $\mathscr{P}:\mathscr{E}\rightarrow\mathscr{E}$ a natural real formally-Hermitian Green-hyperbolic operator, and let $\Pi$ be an assignment (not necessarily functorial), to every object $(M,g)$ in $\mathsf{Loc}$, of a distributional section of $\mathscr{E}^*(M,g)\boxtimes\mathscr{E}^*(M,g)$. Further suppose that $\mathrm{WF}(\Pi)$ satisfies SEP. Let $(M,g)$ be Minkowski spacetime. Then if $\mathrm{WF}(\Pi(M,g))=\{(x,k;x',-k')\in\dot{N}M\times\dot{N}M:(x,k)\sim(x',k')\;\mathrm{and}\; k\in \dot{N}^+M\}$, where $\sim$ means `there exists a geodesic from $x$ to $x'$ with cotangent vector $k$ at $x$ and $k'$ at $x'$', the same is true, at least locally, for any other globally hyperbolic, time-oriented $(M',g')$.

    Moreover, the same is true if we replace $\dot{N}$ with $\dot{J}$ and $\dot{N}^+$ with $\dot{J}^+$, and if $\Pi$ satisfies (antisymmetric part) and (positivity) at each spacetime, the result holds globally.
\end{thm}
Here, $\mathsf{Loc}$ is the category of globally hyperbolic, time-oriented Lorentzian manifolds whose arrows are smooth, causal, orientation-preserving and isometric embeddings; $\mathsf{FB}$ is the category of fibre bundles whose arrows are smooth bundle morphisms; and I refer to \textcite{March2025BPhil,MarchWeatherallcovgauge,MarchWeatherallnaturaltheories,Kolar+etal} for the definitions of a natural bundle and natural operator (which amount, roughly, to the requirement that the field theory at issue can be defined in a uniform and covariant way across all objects in $\mathsf{Loc}$). 

The idea behind this result is as follows. SEP, as explicated by \textcite{March2025BPhil,MarchWeatherallFun1} in the context of classical field theory, says roughly that the dynamics of matter fields depend only on the metric and orientation to first order, in the sense that they can be defined functorially over $\mathsf{Loc}$, and that whenever two spacetimes agree on the metric and orientation to first order at a point (under some local diffeomorphism), the dynamics for matter fields on those spacetimes also agree at that point (under that local diffeomorphism). But this idea can be applied not just to matter dynamics, but rather to any structure defined functorially over $\mathsf{Loc}$, providing that (a) that structure has well-defined pushforwards under local diffeomorphisms, and (b) that structure is defined locally, in that its `value at a point' makes sense. In particular, this is true of the wavefront set of a distribution, since the wavefront set is a subset of the cotangent bundle. To prove Theorem \ref{thm:SEPforhadamard}, one then uses the fact that any two metrics can be made to coincide to first order along some curve (segment) by some (orientation-preserving) local diffeomorphism \parencite{FletcherWeatherallpart1}, since the property of `being tangent to a null geodesic' requires only the first-order structure of the metric. To make the connection to `standard' presentations of the SEP motivation for the Hadamard condition in the literature, one then notes that, for Klein--Gordon theory, the wavefront set of a Hadamard two-point function (and in particular that of the Minkowski vacuum) satisfies $\mathrm{WF}(\omega)=\{(x,k;x',-k')\in\dot{N}M\times\dot{N}M:(x,k)\sim(x',k')\;\mathrm{and}\;k\in \dot{N}^+M\}$. As such, theorem \ref{thm:SEPforhadamard} provides a precise way of making sense of the idea that the Hadamard condition can be motivated via appeal to (a version of) the equivalence principle, applied to the singular structure of the Minkowski vacuum two-point function. This, I take it, is substantial progress on understanding the equivalence principle motivation for the Hadamard condition. 

Theorem \ref{thm:SEPforhadamard} also allows us to considerably sharpen the second worry mentioned above. In the antecedent of Theorem \ref{thm:SEPforhadamard}, we assume both that $\mathrm{WF}(\Pi)$ satisfies SEP, and that $\mathrm{WF}(\Pi)$ agrees with the wavefront set of the Minkowksi vacuum of the Klein--Gordon field at Minkowski spacetime. From the perspective of the ordinary use of the Hadamard condition (e.g., in defining the expectation value of the stress--energy operator, Wick polynomials, etc.), what is important about these assumptions is that they constrain the singular structure of $\Pi$. But from the perspective of SEP, this focus on singular structure is somewhat awkward. This is because $\Pi$ itself cannot satisfy SEP (at least if it is Hadamard at each $(M,g)$). To see this, note that satisfying SEP in the sense which is relevant for Theorem \ref{thm:SEPforhadamard} would require that $\Pi$ be functorial (over $\mathsf{Loc}$). But if $\Pi$ were functorial and Hadamard at each $(M,g)$, then since any Hadamard distribution uniquely determines a corresponding Hadamard state on the algebra $(\mathcal{P}_\mathrm{reg}(\mathscr{E},\mathscr{P})(M,g),\star_{\frac{i}{2}\Delta_P})$, there would exist a functorial assignment of states (for these algebras) to objects in $\mathsf{Loc}$, which is impossible in general \parencite{FewsterVerch2012localitycovariance,Fewster2018states}.\footnote{Strictly speaking, these theorems only establish a conflict between the existence of such a state for a QFT (satisfying various properties), and `non-triviality' of that theory; these conditions are known to be satisfied for standard QFTs of interest, e.g., Klein--Gordon theory. Note also that, for this reason, whilst I am optimistic that new technical work could provide further control over the relationship between SEP and the Hadamard condition---see \textcite[\S 5]{Marchtworoutes} for suggestions in this direction---I do not think that it should be expected to fully resolve the problems with the equivalence principle motivation for the Hadamard condition discussed here.} And if $\Pi$ cannot satisfy SEP whilst agreeing with the two-point function of the Minkowski vacuum of the Klein--Gordon field at Minkowski spacetime, it is not clear why we should expect (or want) $\mathrm{WF}(\Pi)$ to meet these conditions.

There is a kind of flat-footed response to this problem, which is to suggest that perhaps having a wavefront set which satisfies the antecedent of Theorem \ref{thm:SEPforhadamard} is the best one can do by way of making sense of the idea that $\Pi$ `almost' meets the above conditions, or is `as close to satisfying the above conditions as possible'. Perhaps. But this raises a deeper worry: namely, what's so special about singular structure? If the aim of the game is just to find \textit{some} condition along these lines which can be satisfied by the assignment $\Pi$, then the focus on  $\mathrm{WF}(\Pi)$ begins to look \textit{ad hoc}. On the other hand, suppose that we have antecedent reasons for thinking that singular structure \textit{is} somehow special, so that it matters if $\mathrm{WF}(\Pi)$ satisfies these conditions, but we should not care about whether $\Pi$ itself does. That is, one might argue that it is in virtue of satisfying $\mathrm{WF}(\Pi(M,g))=\{(x,k;x',-k')\in\dot{N}\times\dot{N}:(x,k)\sim(x',k')\;\mathrm{and}\; k\in \dot{N}^+\}$ that the Minkowski vacuum two-point function of the Klein--Gordon field has such-and-such desirable properties, and it is these properties which we would like $\Pi$ to reproduce (via SEP). I do not have any objection to this kind of strategy (indeed, I think it is basically correct, and is closely related to the one I will adopt in \S\ref{sec:operatorstohadamard}), but one has to be very careful in pursuing it not to collapse the need for the equivalence principle altogether. That is, if the argument for the assumption that $\mathrm{WF}(\Pi(M,g))=\{(x,k;x',-k')\in\dot{N}\times\dot{N}:(x,k)\sim(x',k')\;\mathrm{and}\; k\in \dot{N}^+\}$ on Minkowski spacetime is supposed to be that it is in virtue of satisfying this condition that the Minkowski vacuum has such-and-such properties, why not apply the same argument \textit{directly} to each globally hyperbolic, time-oriented spacetime $(M,g)$, without the need for the detour via SEP?

Another response would be to suggest that rather than considering a possibly non-functorial assignment of distributional sections of $\mathscr{E}^*\boxtimes\mathscr{E}^*$ to spacetimes, we should consider instead a functorial assignment of a \textit{space} of distributional sections of $\mathscr{E}^*\boxtimes\mathscr{E}^*$ to spacetimes, all of which share the same wavefront set (at each spacetime).\footnote{I am grateful to Jim Weatherall for pressing me on this point.} And indeed this is possible \parencite{Fewster2025Hadamard}. But again, this raises the question why singular structure is `special'---this time in a slightly different form---namely, why it is that the wavefront set of this space of distributional sections is to be held fixed (at each spacetime), rather than anything else?

Theorem \ref{thm:SEPforhadamard} also brings out another worry about the equivalence principle motivation for the Hadamard condition, which has to do with the generality of the approach. Suppose that, for whatever reason, we expect the antecedent of Theorem \ref{thm:SEPforhadamard} to be satisfied in the context of some specific matter theory, e.g., Klein--Gordon theory. It does not follow that we should expect the antecedent of Theorem \ref{thm:SEPforhadamard} to be satisfied for any matter theory of physical interest whatsoever. To give some examples: if $(E,P)$ is a Green hyperbolic field theory on some $(M,g)$ which is $\mathcal{V}^\pm$-decomposable but not functorial (over $\mathsf{Loc}$), or even if it is functorial, but does not satisfy SEP, then there is no reason to expect that $\mathrm{WF}(\Delta_P)$ will satisfy SEP, and \textit{a fortiori} that $\mathrm{WF}(\Pi)$ will do so (if $\Pi$ is Hadamard for $P$). One might reply that the class of theories for which we do expect the antecedent of Theorem \ref{thm:SEPforhadamard} to hold is precisely the class of theories we should consider `physically reasonable', and it is in these cases that the Hadamard condition has bite. Quite apart from the fact that it is not clear that all theories which we do think of as `physically reasonable' (e.g., Yang--Mills theory) fall into this class, I will just note that if this is right, then the sense in which Hadamard condition is a necessary condition on `physically reasonable' states of the quantum field is rather different---and more attenuated---than the way it is standardly presented in the literature. 

\subsection{Pragmatic considerations}\label{sec:itsallpragmaticsisntit}

This takes us on to a second family of possible motivations for the Hadamard condition, which are more pragmatic in character. On this approach, one begins by noting that there are many theoretical contexts in which it is incredibly useful to impose the Hadamard condition. To give a non-exhaustive list (and to recap \S\ref{sec:intro}): Hadamard states permit the definition of the expectation value of the stress--energy operator, which is needed in semi-classical gravity (where it enters into the left-hand side of the semi-classical Einstein's equation); relatedly, the Hadamard condition is used in derivations of quantum energy inequalities; the Hadamard condition is also important in derivations of Hawking radiation and the Hawking temperature, where it enters into the definition of the relevant `vacuum-like' state of a black hole spacetime; the Hadamard condition permits the definition of Wick polynomials and (in turn) time-ordered products and the perturbative S-matrix on curved spacetime, etc. One then argues that the Hadamard condition is motivated---and should be thought of as a necessary condition on `physically reasonable' states of the quantum field---precisely by the fact that the condition, e.g., ensures well-definedness of these quantities, permits these derivations, etc.

Motivations for the Hadamard condition with this kind of flavour are also abundant in the literature. An early example can be found in \textcite{Fullingetal1978}, who write:
\begin{quote}
    All these considerations [having to do with the possibility of renormalizing the stress--energy operator] suggest that the validity of [the Hadamard condition] be regarded as a basic criterion for a ``physically reasonable" state, perhaps even as the definition of that phrase. \parencite[259]{Fullingetal1978}
\end{quote}
A similar idea is picked up by \textcite{KayWald1991}:
\begin{quote}
    It turns out that the set of all quasifree states is considerably too large in that it includes, in particular, states for which the ``expectation value of the energy—momentum tensor operator" is infinite or indefinable. We shall adopt, here, the viewpoint [...]\ that the only physically acceptable states are the ``Hadamard states" [...]. For such states, the expectation value of the stress—energy operator can be defined in a completely satisfactory manner via the point-splitting prescription [...]. \parencite[82]{KayWald1991}
\end{quote}

This line of argument also continues into the contemporary literature on the Hadamard condition. Here are two representative passages:
\begin{quote}
    For many purposes, it can be argued that Hadamard states form the preferred class of physical states on a linear scalar QFT [...]. In particular, it is for this class of states that the expectation value of non-linear observables, such as the (renormalized) stress–-energy tensor can be evaluated. \parencite[1--2]{JanssenVerch2023}
\end{quote}
\begin{quote}
    Hadamard states play an important role in computations, for they permit the evaluation of Wick polynomials, including the stress--energy tensor, and time-ordered products [...]. [M]any results can be proved for general Hadamard states on general spacetime backgrounds. Examples include quantum energy inequalities [...], no-go results concerning chronology violating spacetimes [...] and existence results for the semiclassical Friedmann equations [...]. Hadamard states also play an important role in understanding why black holes display the Hawking temperature [...]. \parencite[2]{FewsterVerch2013}
\end{quote}

I do not have any objection to this kind of strategy \textit{per se}, and my preferred answer to the question `why Hadamard states?'\ will also have a lot to do with well-definedness of certain quantities that we take to be `physical'. But even so, it would be desirable to have some clarity on what, exactly, it is that the condition allows us to do which we could not do otherwise, and how, exactly, this is supposed to underwrite the claim that `physically reasonable' states of the quantum field are necessarily Hadamard, rather than simply enumerating the various theoretical applications in which the condition has proved useful. As an example, consider the fact that the expectation value of the stress-energy operator (e.g., of the Klein--Gordon field) is well-defined for Hadamard states, whereas it is not well-defined in general \parencite{Wald1977,Wald1978,Wald1994}. Does this mean that only Hadamard states are `physically reasonable', since the expectation value of the stress-energy operator ought to be well-defined, \textit{qua} physical quantity? Or does it mean that we need a different understanding of what the relevant `physical' quantities are (or different mathematics to describe those quantities)?

This worry can be sharpened along a few different directions. First, typically, discussions of the relationship between Hadamard states and the definition of the expectation value of the stress--energy operator, Wick polynomials, etc., only establish that the Hadamard condition is \emph{sufficient} for well-definedness of these quantities, whereas presumably, what one would like in order to underwrite the claim that `physically reasonable' states of linear QFTs are necessarily Hadamard is an argument that the Hadamard condition is also \emph{necessary} for well-definedness of these quantities. Second, the claim that the expectation value of the stress--energy operator ought to be well-defined is often tied to this being a precondition for doing semi-classical gravity \parencite{Wald1977,Wald1978,Wald1994,Fullingetal1978}.\footnote{I am grateful to an anonymous referee for pressing me on this point.} But semi-classical gravity is known to suffer from various problems (e.g., (in)famously, the black hole information-loss paradox) which may or may not be artefacts the semi-classical regime, and so it is not clear that the lack of a well-defined stress--energy expectation value for arbitrary states should not simply be chalked up to this as well. The fact that the Hadamard condition also plays a role in defining Wick polynomials, time-ordered products, etc., on curved spacetime backgrounds is certainly an important step forward in making the case that the need for Hadamard states should not be understood as an artefact of the semi-classical regime in this way, but even so, one might ask for something stronger---e.g., arguments to the effect that a mathematically well-defined quantum field-theoretic treatment of any physical quantity with such-and-such general properties will require the Hadamard condition. Third, it is unclear how these kind of pragmatic considerations are supposed to support the intended interpretation of Hadamard states as `vacuum-like'---to which derivations making use of the condition often appeal. Somewhat more generally: if one of the worries with which we ended the last section is that the motivation via the equivalence principle had too little to do with the ways in which the Hadamard condition is applied in practice, the worry here is that this motivation risks becoming too close to the ways in which the Hadamard condition is applied in practice to underwrite the physical and theoretical significance that these applications of condition are supposed to have.

\section{From operator products to the Hadamard condition}\label{sec:operatorstohadamard}

Let's take stock. On the one hand, the equivalence principle motivation makes reasonably good sense of the idea that Hadamard states are supposed to be `vacuum-like', but struggles when it comes to making the connection to the reasons for the ordinary use of the Hadamard condition in QFT on curved spacetime practice, and one might also worry about the generality of the approach. On the other hand, the pragmatic motivation is tightly connected to the reasons for the ordinary use of the Hadamard condition in QFT on curved spacetime practice, but also suffers from some concerns about its clarity and generality---in particular, in articulating the sense in which the Hadamard condition is supposed to be a \emph{necessary} (rather than sufficient) condition on physically reasonable states of any linear quantum field---and doesn't seem to underwrite the idea that Hadamard states are supposed to be `vacuum like'. Combining both motivations would help to fill in some of these gaps, but not all. Can we do better?

I think we can. To get to this, I want to return to our discussion of operator products in \S\ref{sec:observables}. There, I noted that any distributional section of $E^*\boxtimes E^*$ defines an operator product on $\mathcal{P}_\mathrm{reg}(E)[[\hbar]]$ via \eqref{eq:starprod}. But the same is not true on $\mathcal{P}_{\mathcal{V}\mathrm{mc}}(E)[[\hbar]]$---since some of the coefficients of the observables in $\mathcal{P}_{\mathcal{V}\mathrm{mc}}(E)[[\hbar]]$ are singular distributions, one has to pay attention to when H\"ormander's criterion for the existence of the products of distributions appearing in the expansion of \eqref{eq:starprod} is met. Moreover, in building a quantum field theory, not just any operator product will do---one would also like this product e.g., to be well-behaved under complex conjugation, to bear the right kind of relationship to the dynamics of the theory, and so on. Continuing along these lines, as I show in \parencite{Marchtworoutes}, we have the following equivalence result:
\begin{thm}[March 2026]\label{thm:operatorstohadamard}
    Let $E$ be a Hermitian vector bundle, $\Pi$ a distributional section of $E^*\boxtimes E^*$, and $P$ a $\mathcal{V}^\pm$-decomposable real formally-Hermitian Green-hyperbolic operator on $E$. The following conditions are equivalent:
    \begin{itemize}
        \item $\Pi$ is a (bosonic) $\mathcal{V}^+$-Hadamard or $\mathcal{V}^-$-Hadamard two-point function for $P$.
        \item The unital associative algebra product $\star_{\Pi}$ induced by $\Pi$ is well-defined on $\mathcal{P}_{\mathcal{V}\mathrm{mc}}(E)[[\hbar]]$, in the sense that H\"ormander's criterion is satisfied, and has the following properties:
        \begin{enumerate}
            \item $\star_\Pi$ is a star product on $\mathcal{P}_{\mathcal{V}mc}(E)[[\hbar]]$;
            \item $\star_\Pi$ descends to $\mathcal{P}_{\mathcal{V}mc}(E,P)[[\hbar]]$;
            \item $\star_\Pi$ is a deformation quantization of $[\,\cdot\,,\,\cdot\,]_P$, i.e., as $\hbar\rightarrow 0$, $A_1\star_{\Pi}A_2\rightarrow \mathrm{prod}(A_1\otimes A_2)$ and $(A_1\star_{\Pi}A_2-A_2\star_{\Pi}A_1)/i\hbar\rightarrow [A_1,A_2]_{P}$; and
            \item The map $\langle\,\cdot\,\rangle_\mathbf{0}:\mathcal{P}_{\mathcal{V}mc}(E)[[\hbar]]\rightarrow \mathbb{C}[[\hbar]]$, $\langle A\rangle_\mathbf{0}:=A(\mathbf{0})$ is a state on the algebra $(\mathcal{P}_{\mathcal{V}mc}(E)[[\hbar]],\star_\Pi)$.
        \end{enumerate}
    \end{itemize}
\end{thm}
(It is worth remarking that this result extends to fermionic field theories under appropriate modifications, the details of which would take us a little too far afield here---see \textcite[Appendix A]{Marchtworoutes}.) In other words, any unital associative operator product on $\mathcal{P}_{\mathcal{V}\mathrm{mc}}(E)[[\hbar]]$ satisfying conditions (1)--(4) is necessarily a $\mathcal{V}^\pm$-Hadamard two-point function, and conversely, any $\mathcal{V}^\pm$-Hadamard two-point function defines an operator product meeting these conditions. 

Theorem \ref{thm:operatorstohadamard} is, I claim, best read as a substantial strengthening, clarification, and generalization of the pragmatic motivation for the Hadamard condition discussed in \S\ref{sec:itsallpragmaticsisntit}. Getting to this---and showing how Theorem \ref{thm:operatorstohadamard} avoids the worries raised in \S\ref{sec:why} about both the standard presentations of the pragmatic motivation in the literature and the equivalence principle motivation---will occupy us for the remainder of this section. To this end, I will now discuss each of the conditions in Theorem \ref{thm:operatorstohadamard} in turn. 

\subsection{The `structure of a QFT' conditions (1)--(3)}

I will begin with conditions (1)--(3). These, I suggest, are minimal adequacy conditions that one should expect any `physically reasonable' QFT to satisfy. In some detail: condition (1)---that the algebra $(\mathcal{P}_{\mathcal{V}\mathrm{mc}}(E)[[\hbar]],\star_\Pi)$ be a star algebra---is just what it means for the algebra $(\mathcal{P}_{\mathcal{V}\mathrm{mc}}(E)[[\hbar]],\star_\Pi)$ to be equipped with complex conjugation, i.e., to be a \emph{complex} algebra of \emph{complex-valued} observables. Condition (2) says that we can consider the $(\mathcal{P}_{\mathcal{V}\mathrm{mc}}(E,P)[[\hbar]],\star_\Pi)$ as an algebra in its own right. This captures a minimal sense in which the algebraic structure of a QFT `respects' the dynamics of that QFT, insofar as that algebraic structure makes sense under the restriction on-shell. It is also presupposed by any approach to QFT in which we wish to be able to make use of (faithful) Fock space representations in order to perform computations. 

Condition (3) is, plausibly, just what it means for the quantum field theory $(\mathcal{P}_{\mathcal{V}\mathrm{mc}}(E,P)[[\hbar]],\star_\Pi)$ to be a quantization of the classical field theory $(E,P)$. One might object to (3) on the grounds that requiring the QFT $(\mathcal{P}_{\mathcal{V}\mathrm{mc}}(E)[[\hbar]],\star_\Pi)$ to be a deformation quantization of of the classical field theory $(E,P)$ rests on a conception of classical field theory as somehow `prior to' QFT, or of QFT as somehow `parasitic upon' classical field theory. There are, I think, two things to say here. One is that, however nice it would be to have a self-contained definition of what a QFT is that eschewed any reference to classical field-theoretic (or indeed, classical mechanical) structures whatsoever, we do not have that (yet). (Every approach to defining a QFT that I know of makes use of Poisson bracket structures at some point in the analysis, though in some approaches (e.g., those beginning with the construction of the Klein--Gordon Fock space from modes of the quantum harmonic oscillator) this is buried deeper in the construction than in others.) The second, and related, point is that it remains the case that the \emph{Lagrangian} (and hence, the Euler--Lagrange operator) of a QFT is perfectly well-defined as a classical object (at least if it is understood in the sense of generalized Lagrangians, see \textcite{Brunettietal2019}), and that our understanding of what it means for a QFT to be a QFT \emph{for that Euler--Lagrange operator} is in fact tied up with the operator product of that theory bearing a certain relationship to Poisson bracket-like structures definable from that Euler--Lagrange operator. The term `deformation quantization' is a convenient way to describe that relationship, but it need not be taken to imply that the classical field theory is primary to the quantum one.

\subsection{The `vacuum-like states' condition (4)}\label{sec:vacuum-like}

Condition (4) is more interesting. (4) is the requirement that the algebra $(\mathcal{P}_{\mathcal{V}\mathrm{mc}}(E,P)[[\hbar]],\star_\Pi)$ admit states which are `analogues' of the classical vacuum field configuration, in the sense that these states consist in evaluating all observables at the zero section of $E$ (i.e., the classical solution of $P$ in which all matter fields vanish). The first point to note about this is that it substantially clarifies the sense in which Hadamard states are `vacuum-like', though in a slightly different way from what one might have expected. In particular, (4) invites us to consider Hadamard states not, in the first instance, as a generalization of the Minkowski vacuum state, but as quantum `analogues' of the classical vacuum field configuration (which are then related to the Minkowski vacuum insofar as the Minkowski vacuum is also such a state). But, I suggest, this makes the sense in which Hadamard states are `vacuum-like' more transparent than the standard presentation of Hadamard states as direct generalizations of the Minkowski vacuum state. The classical vacuum field configuration can be defined covariantly (i.e., functorially) over all globally-hyperbolic, time-oriented spacetimes $(M,g)$ whenever $E$ can also be so defined, so we can control over what a `quantum analogue' of the classical vacuum field configuration should look like at any globally-hyperbolic, time-oriented spacetime $(M,g)$ in a uniform way, whereas we cannot seek a covariant extension of the Minkowski vacuum to every such $(M,g)$ (since no such states exist \parencite{FewsterVerch2012localitycovariance,Fewster2018states}).\footnote{Continuing along these lines, one might also suggest that, in the context of QFT on curved spacetime, there would also be something unsatisfactory about a motivation for the Hadamard condition which made ineliminable reference to QFT on Minkowski spacetime, since the latter is supposed to be a derivative application of the former. I am sympathetic to this idea. That said, there are of course approaches to the foundations of general relativity which do regard Minkowski spacetime as having a privileged role to play in the interpretation of that theory (see e.g., \textcite{Brown2005,LinnemannRead}), whereas the argument in the main text can stand apart from these issues.}  As such, Theorem \ref{thm:operatorstohadamard} avoids the worry mentioned in \S\ref{sec:itsallpragmaticsisntit} about how the pragmatic motivation for the Hadamard condition is supposed to underwrite the intended interpretation of Hadamard states as `vacuum-like'. 

The second point is that (4) has a rather different status from (1)--(3). As argued above, (1)--(3) are plausibly minimal adequacy conditions that one should expect (some version of) to be satisfied any `physically reasonable' QFT. But (4) is not. Indeed, it is straightforward to come up with examples of QFTs in which one would not automatically expect something like (4) to hold---e.g., interacting field theories, or (as a particular example thereof) field theories whose classical ground state configurations exhibit spontaneously-broken symmetry.\footnote{I am grateful to David Wallace for suggesting this example.} The right thing to say here, I think, is that neither of these contexts is presently under consideration. By assumption (and this is the case in all contexts where discussions of the Hadamard are carried out), the field theory $(E,P)$ under consideration is linear---that is, $P$ is the operator of a \emph{free} field theory. One can, of course, treat interacting theories (perturbatively) in this framework, working in the interaction picture.\footnote{To be clear, in order to avoid getting into Haag's theorem-type issues, the term `interaction picture' should be understood here in the formal sense of constructing an `interacting' algebra of observables in terms of a suitable transformation defined on the free algebra.} But as a starting point for carrying out that analysis, one must first construct the free theory. And plausibly, admitting states which are vacuum-like in \emph{some} sense is an adequacy condition on what it is for a QFT to be free.

That said, this kind of argument only suffices to establish that a `physically reasonable' linear QFT should \emph{admit} `vacuum-like' states \emph{per} condition (4). To recover the claim that `physically reasonable' states of a linear QFT are necessarily Hadamard, one would require an argument for the (stronger) claim that `physically reasonable' states of such QFTs must be `vacuum-like' in this sense. I will return to this at the end of \S\ref{sec:whymicrocausal}.

\subsection{The space of observables $\mathcal{P}_{\mathcal{V}mc}(E)[[\hbar]]$}\label{sec:whymicrocausal}

Finally, we come to the space of observables $\mathcal{P}_{\mathcal{V}mc}(E)[[\hbar]]$, which will allow us to make contact with standard presentations of the pragmatic motivation for the Hadamard condition discussed in \S\ref{sec:itsallpragmaticsisntit}. Why require well-definedness of $\star_\Pi$ on $\mathcal{P}_{\mathcal{V}\mathrm{mc}}(E)[[\hbar]]$? In fact, there are really two questions here, which involve rather different answers. First, why not consider a larger space, such as $\mathcal{P}(E)[[\hbar]]$? And second, why not consider a smaller space, such as $\mathcal{P}_{\mathrm{reg}}(E)[[\hbar]]$? 

Beginning with the former, this comes back to condition (3). Since $i\Delta_P=\Pi-{}^t\Pi$ requires, in particular, that $\mathrm{WF}(\Delta_P)\subset\mathrm{WF}(\Pi)\cup\mathrm{WF}({}^t\Pi)$, we cannot make $\mathrm{WF}(\Pi)$ smaller than what is imposed by WFSSC (as working on a larger space of observables with than $\mathcal{P}_{\mathcal{V}\mathrm{mc}}(E)[[\hbar]]$ would require) whilst maintaining (3). So given that (3) is plausibly just what it means for $(\mathcal{P}_{\mathcal{V}\mathrm{mc}}(E)[[\hbar]],\star_\Pi)$ to be a quantization of $(E,P)$, we cannot use a larger space of observables than $\mathcal{P}_{\mathcal{V}\mathrm{mc}}(E)[[\hbar]]$.

Turning now the latter, ultimately, I take this to be a question of what observables the practice of our most empirically-successful mathematical and theoretical physics requires us to consider. But here, I take the case in favour of $\mathcal{P}_{\mathcal{V}\mathrm{mc}}(E)[[\hbar]]$ to be fairly decisive. To see this, note that field amplitude polynomials routinely considered in QFT calculations such as $\varphi^2$, $\varphi^4$, etc., are $\mathcal{V}$-microcausal polynomial observables (for any $\mathcal{V}^\pm$), with e.g., for $\varphi^2$, $\alpha^{(2)}(x,y)=\delta(x,y)\langle\,\cdot\,,\cdot\,\rangle$, etc., and $\mathrm{WF}(\delta(x,y)\langle\,\cdot\,,\cdot\,\rangle)=\{(x,k_x;y,k_y)\in \dot{T}^*M\times \dot{T}^*M: x=y, k_y=-k_x\}$. As a second example, consider the stress--energy operator. This is, in particular, a non-linear observable constructed out of polynomials of fields and their derivatives (e.g., for Klein--Gordon theory, $T^{ab}(\varphi)=\nabla^a\varphi\nabla^b\varphi-1/2g^{ab}(\nabla_n\varphi\nabla^n\varphi-m^2\varphi^2)$). Fix a (local) coordinate system on some $U\subset M$. Any component $T^{\mu\nu}$ of $T^{ab}$ relative to this coordinate system is a non-linear \emph{local observable}, i.e., a functional $A:\Gamma_c(E)\rightarrow \mathbb{C}$ such that $A(\varphi)=\int_M\alpha(x,f(x),\partial f(x),...,\partial^kf(x))\mathrm{dvol}_g(x)$ for some smooth $\alpha(x,f(x),\partial f(x),...,\partial^kf(x))$ and some finite $k$; that is, the stress--energy operator is a local observable-valued tensor. Three important facts about local observables are as follows \parencite{FredenhagenRejzner2015}. First, any local observable is a polynomial observable, whose coefficients $\alpha^{(i)}$ are such that $\mathrm{WF}(\alpha^{(i)})\subseteq\{(x_1,k_1;...;x_i,k_i)\in \times^iT^*M:x_1=...=x_i\;\mathrm{and}\;\sum^i_{j=0}k_j=\mathbf{0}\}$. Thus local observables are $\mathcal{V}$-microcausal for any $\mathcal{V}^\pm$. Second, a local observable is a regular observable only if it is linear. Third, the local observables form a dense subspace (in the H\"ormander topology---see e.g., \parencite{FredenhagenRejzner2015}) of the $\mathcal{V}$-microcausal polynomial observables for any $\mathcal{V}^\pm$. Thus, requiring a well-defined operator product on all the closure of the space of all local observables (in the H\"ormander topology) is equivalent to requiring a well-defined operator product on $\mathcal{P}_{\mathcal{V}\mathrm{mc}}(E)[[\hbar]]$. 

This is not (quite) the usual arguments about well-definedness of the expectation value of the stress--energy operator, Wick polynomials, etc.\ mentioned in \S\ref{sec:itsallpragmaticsisntit}, however. To make direct contact with those arguments, we need one further ingredient. Suppose one also requires (1)--(4), so that $\Pi=\Delta_H$ is Hadamard. In many contexts, it would be preferable to work with $\star_{\frac{i}{2}\Delta_P}$ rather than $\star_{\Delta_H}$. In particular, $\star_{\frac{i}{2}\Delta_P}$ permits covariant definition of a QFT over the class of globally-hyperbolic, time-oriented spacetimes whenever $P$ can be defined covariantly over that same class of spacetimes, whereas $\star_{\Delta_H}$ does not \parencite{FewsterVerch2012localitycovariance,Fewster2018states}. But, as we have just seen, if we would like the algebra of observables of the quantum field to contain quantities such as the stress--energy operator, non-linear field polynomials (taken pointwise), etc., then by theorem \ref{thm:operatorstohadamard}, it seems like $\star_{\Delta_H}$ is forced upon us. So what gives?

Here is one option. We could work primarily with the product $\star_{\frac{i}{2}\Delta_P}$ on $\mathcal{P}_{\mathrm{reg}}(E)[[\hbar]]$, and then use the algebra $(\mathcal{P}_{\mathcal{V}\mathrm{mc}}(E)[[\hbar]],\star_{\Delta_H})$ to treat any remaining elements of $\mathcal{P}_{\mathcal{V}\mathrm{mc}}(E)[[\hbar]]$ that we need on a case-by-case basis. How would this work? Consider the map $H:\mathcal{P}_{\mathrm{reg}}(E)[[\hbar]]\rightarrow\mathcal{P}_{\mathrm{reg}}(E)[[\hbar]]$ defined by $H(A)=\mathrm{exp}(\frac{1}{2}(\Delta_H-\frac{i}{2}\Delta_P))\circ A$. One can verify that this is an isomorphism of star algebras $(\mathcal{P}_{\mathrm{reg}}(E)[[\hbar]],\star_{\frac{i}{2}\Delta_P})$ and $(\mathcal{P}_{\mathrm{reg}}(E)[[\hbar]],\star_{\Delta_H})$. But $H^{-1}$ can be defined on $\mathcal{P}_{\mathcal{V}\mathrm{mc}}(E)[[\hbar]]$. So, using $H^{-1}$, one can supply the `missing' elements of $(\mathcal{P}_{\mathrm{reg}}(E)[[\hbar]],\star_{\frac{i}{2}\Delta_P})$ that one needs from $(\mathcal{P}_{\mathcal{V}\mathrm{mc}}(E)[[\hbar]],\star_{\Delta_H})$. On $\mathcal{P}_{\mathrm{reg}}(E,P)[[\hbar]]$, $H^{-1}$ corresponds to normal ordering with respect to the state $\langle\,\cdot\,\rangle_\mathbf{0}$. On $\mathcal{P}_{\mathcal{V}\mathrm{mc}}(E,P)[[\hbar]]$, $H^{-1}$ corresponds to taking the limit (in the H\"ormander topology) of a sequence of normally-ordered field polynomials smeared by sequences of functions which converge (in the H\"ormander topology) to multivariable delta distributions, which is just the usual `point-splitting' prescription for the evaluation of the stress--energy expectation value, Wick polynomials, etc.\ with respect to a Hadamard state discussed in e.g., \parencite{Fullingetal1978,Wald1994,HollandsWald2001}. The fact that $(\mathcal{P}_{\mathcal{V}\mathrm{mc}}(E)[[\hbar]],\star_{\Delta_H})$ is an algebra meeting conditions (1)--(3) guarantees that the basic structural features one expects from a QFT are preserved under this prescription. Thus we see how the pragmatic motivation for the Hadamard condition discussed in \S\ref{sec:itsallpragmaticsisntit} can be recovered by Theorem \ref{thm:operatorstohadamard}, but now in a considerably generalized form where one can also see that the Hadamard condition is \emph{necessary} for that argument to go through, and for a large class of QFTs going well beyond the example of Klein--Gordon theory standardly used in the literature. As such, Theorem \ref{thm:operatorstohadamard} is tightly connected to the reasons the Hadamard condition is imposed in practice, but also gives us sufficient control over what the Hadamard condition is doing in these applications to underwrite the physical and foundational significance the condition is supposed to have.

This also allows us to address the question with which we ended \S\ref{sec:vacuum-like}, i.e., whether the best argument for condition (4) is that we should expect `physically-reasonable' linear QFTs to \emph{admit} `vacuum-like' states \emph{per} (4), or whether some argument could be found for the stronger claim that `physically reasonable' states of `physically reasonable' linear QFTs must be `vacuum-like' in this sense. Suppose one adopts the above strategy of taking $(\mathcal{P}_{\mathrm{reg}}(E)[[\hbar]],\star_{\frac{i}{2}\Delta_P})$ as primary and using $H^{-1}$ to supply any remaining elements of $\mathcal{P}_{\mathcal{V}\mathrm{mc}}(E)[[\hbar]]$ that one needs case-by-case. Hadamard two-point functions are not unique, and from the perspective of the algebra $(\mathcal{P}_{\mathrm{reg}}(E)[[\hbar]],\star_{\frac{i}{2}\Delta_P})$, no one choice of Hadamard two-point function, an hence no choice of $H^{-1}$, is privileged over any other (though all such choices give rise to isomorphic, albeit unitarily inequivalent, algebras---see e.g., \textcite{FredenhagenRejzner2015,HollandsWald2001}). But if the state of the quantum field one is considering is Hadamard, then one has a canonical choice to go with (at least in the context in which one is working). Thus, if one requires that a `physically reasonable' state of a QFT be one in which the construction of quantities such as the expectation value of the stress--energy operator, Wick polynomials, etc., is canonical (for that state), then one has an argument for the stronger claim.

\section{Close}\label{sec:conclusion}

The Hadamard condition plays a central role in QFT in curved spacetime, locally covariant QFT, and semi-classical gravity. Providing an adequate motivation for the Hadamard condition is an important task for any philosophical or foundational inquiry into these theories, as well as for any future quantum gravity theory that we might hope to build off the back of them. Yet, as I have argued in this paper, the two standard motivations for the Hadamard condition found in the literature suffer from gaps. In this paper, I have explained how the technical work carried out in \parencite{Marchtworoutes} can be used to fill those gaps. I have also explained how the results of \parencite{Marchtworoutes} can be used to clarify the relationship between the equivalence principle and the Hadamard condition, and discussed the limitations of these results understood as providing a motivation for the Hadamard condition via the equivalence principle.

In undertaking this work, I also hope to have paved the way for a more thoroughgoing research programme surrounding the foundations, epistemology, and methodology of the Hadamard condition. To end this paper, I wish to mention just a few examples. First, one might hope to bring new technical work to bear on better understanding the foundations of the Hadamard condition---in particular, to further clarify the relationship between the Hadamard condition and the equivalence principle, as discussed in \textcite[\S 5]{Marchtworoutes}, and in relating the present paper to the argument for the Hadamard condition given in \textcite{FewsterVerch2013}. Second, the discussion of motivations for the Hadamard condition in this paper has been carried out primarily on the theoretical side. But one might also be interested in a more epistemological angle on the question `why Hadamard states?'---i.e., should we believe that (in certain regimes, to certain degrees of approximation) the Hadamard condition is in fact satisfied by states of the quantum field in the actual world? Third, one might also wish to undertake a more methodology-forward study of the Hadamard condition, seeking better to understand, classify, and characterize the ways in which the Hadamard condition is used in QFT on curved spacetime practice.
These projects will have to wait for another day.

\section*{Acknowledgments}
This material is based upon work supported by the National Science Foundation under Grant No.\ 2419967. I would like to thank Jeff Barrett, India Bhalla-Ladd, Silvester Borsboom, Chris Fewster, Frank Hu, James Read, Dominic Ryder, David Wallace, Jim Weatherall, two anonymous referees, and audiences in Lake Bohinj and Irvine for helpful comments and discussions. 

\printbibliography

@book{Hoermander1,
    author = {Lars H\"ormander},
    title = {The Analysis of Linear Partial Differential Operators I: Distribution Theory and Fourier Analysis},
    publisher = {Springer},
    year = {2003},
    place = {Berlin, Heidelberg},
    doi = {10.1007/978-3-642-61497-2}
}

@article{Fewster2025Hadamard,
    author = {Fewster, Christopher J.},
    title = {Hadamard States for Decomposable Green-Hyperbolic Operators},
    journal = {Communications in Mathematical Physics},
    year = {2025},
    volume = {407},
    number = {1},
    doi = {10.1007/s00220-025-05512-1}
}

@article{FewsterVerch2012localitycovariance,
    author = {Fewster, Christopher J. and Verch, Rainer},
    title = {Dynamical Locality and Covariance: What Makes a Physical Theory the Same in all Spacetimes?},
    journal = {Annales Henri Poincar\'e},
    year = {2012},
    pages = {1613--1674},
    volume = {13},
    number = {7},
    doi = {10.1007/s00023-012-0165-0}
}

@article{Fewster2018states,
    author = {Fewster, Christopher J.},
    title = {The Art of the State},
    journal = {International Journal of Modern Physics D},
    volume = {27},
    number = {11},
    year = {2018},
    doi = {10.1142/S0218271818430071}
}

@misc{March2025BPhil,
    author = {Eleanor March},
    title = {Minimal Coupling, the Strong Equivalence Principle, and the Adaptation of Matter to Spacetime Geometry},
    year = {2025},
    note = {BPhil Thesis, University of Oxford},
    doi = {10.5287/ora-agd8vvyd8}
}

@article{FletcherWeatherallpart1,
    author = {Samuel C. Fletcher and James Owen Weatherall},
    title = {The Local Validity of Special Relativity, Part 1: Geometry},
    journal = {Philosophy of Physics},
    year = {2023},
    volume = {1},
    number = {1},
    doi = {10.31389/pop.6}
}

@article{FletcherWeatherallpart2,
    author = {Samuel C. Fletcher and James Owen Weatherall},
    title = {The Local Validity of Special Relativity, Part 2: Matter Dynamics},
    journal = {Philosophy of Physics},
    year = {2023},
    volume = {1},
    number = {1},
    doi = {10.31389/pop.7}
}

@book{Brown2005,
    author = {Harvey R. Brown},
    title = {Physical Relativity: Space-time structure from a dynamical perspective},
    publisher = {OUP},
    year = {2005},
    place = {Oxford},
    doi = {10.1093/0199275831.001.0001}
}

@article{Readetal2018,
    title = {Two miracles of general relativity},
    journal = {Studies in History and Philosophy of Science Part B: Studies in History and Philosophy of Modern Physics},
    volume = {64},
    pages = {14--25},
    year = {2018},
    doi = {10.1016/j.shpsb.2018.03.001},
    author = {James Read and Harvey R. Brown and Dennis Lehmkuhl},
}

@article{GhinsBudden2001,
    title = {The Principle of Equivalence},
    journal = {Studies in History and Philosophy of Science Part B: Studies in History and Philosophy of Modern Physics},
    volume = {32},
    number = {1},
    pages = {33-51},
    year = {2001},
    doi = {https://doi.org/10.1016/S1355-2198(00)00038-1},
    author = {Michel Ghins and Tim Budden}
}

@article{MarchWeatherallnaturaltheories,
    author = {Weatherall, James Owen and March, Eleanor},
    title = {Natural Theories},
    journal = {The British Journal for the Philosophy of Science},
    year = {Forthcoming},
    doi = {10.1086/740611}
}

@article{MarchWeatherallcovgauge,
    author = {March, Eleanor and Weatherall, James Owen},
    title = {A puzzle about general covariance and gauge},
    journal = {The British Journal for the Philosophy of Science},
    year = {Forthcoming},
    doi = {10.1086/736478}
}

@book{Wald1994,
    author = {Robert M. Wald},
    title = {Quantum Field Theory in Curved Spacetime and Black Hole Thermodynamics},
    publisher = {University of Chicago Press},
    year = {1994},
    place = {Chicago, IL}
}

@article{Wald1977,
    author = {Robert M. Wald},
    title = {The back reaction effect in particle creation in curved spacetime},
    journal = {Communications in Mathematical Physics},
    year = {1977},
    volume = {54},
    pages = {1--19},
    number = {1},
    doi = {10.1007/BF01609833}
}

@article{Wald1978,
    title = {Trace anomaly of a conformally invariant quantum field in curved spacetime},
    author = {Wald, Robert M.},
    journal = {Physical Review D},
    volume = {17},
    number = {6},
    pages = {1477--1484},
    year = {1978},
    doi = {10.1103/PhysRevD.17.1477}
}

@article{KayWald1991,
    title = {Theorems on the uniqueness and thermal properties of stationary, nonsingular, quasifree states on spacetimes with a bifurcate killing horizon},
    journal = {Physics Reports},
    volume = {207},
    number = {2},
    pages = {49--136},
    year = {1991},
    doi = {10.1016/0370-1573(91)90015-E},
    author = {Bernard S. Kay and Robert M. Wald},
}

@article{JanssenVerch2023,
    doi = {10.1088/1361-6382/acb039},
    year = {2023},
    volume = {40},
    number = {4},
    author = {Janssen, Daan W. and Verch, Rainer},
    title = {Hadamard states on spherically symmetric characteristic surfaces, the semi-classical Einstein equations and the Hawking effect},
    journal = {Classical and Quantum Gravity}
}

@article{Fullingetal1978,
    author = {Fulling, Stephen A. and Sweeny, Mark and Wald, Robert M.},
    title = {Singularity structure of the two-point function in quantum field theory in curved spacetime},
    journal = {Communications in Mathematical Physics},
    year = {1978},
    pages = {257--264},
    volume = {63},
    number = {3},
    doi = {10.1007/BF01196934}
}

@article{FewsterVerch2013,
    doi = {10.1088/0264-9381/30/23/235027},
    year = {2013},
    volume = {30},
    number = {23},
    author = {Fewster, Christopher J. and Verch, Rainer},
    title = {The necessity of the Hadamard condition},
    journal = {Classical and Quantum Gravity}
}

@article{Amir-HomayoonOttewill,
    title = {Quantum states and the Hadamard form. III. Constraints in cosmological space-times},
    author = {Najmi, Amir-Homayoon and Ottewill, Adrian C.},
    journal = {Physical Review D},
    volume = {32},
    number = {8},
    pages = {1942--1948},
    year = {1985},
    doi = {10.1103/PhysRevD.32.1942},
}

@inbook{KhavkineMoretti2015,
    author = {Khavkine, Igor and Moretti, Valter},
    editor = {Brunetti, Romeo and Dappiaggi, Claudio and Fredenhagen, Klaus and Yngvason, Jakob},
    title = {Algebraic QFT in Curved Spacetime and Quasifree Hadamard States: An Introduction},
    booktitle = {Advances in Algebraic Quantum Field Theory},
    year = {2015},
    publisher = {Springer International Publishing},
    address = {Cham},
    pages = {191--251},
    doi = {10.1007/978-3-319-21353-8_5}
}

@book{Duetsch,
    author = {Michael D\"utsch},
    title = {From Classical Field Theory to Perturbative Quantum Field Theory},
    publisher = {Birkh\"auser},
    place = {Cham},
    year = {2019},
    doi ={10.1007/978-3-030-04738-2}
}

@inbook{FredenhagenRejzner2015,
    author="Fredenhagen, Klaus and Rejzner, Katarzyna",
    editor="Calaque, Damien and Strobl, Thomas",
    title="Perturbative Algebraic Quantum Field Theory",
    booktitle="Mathematical Aspects of Quantum Field Theories",
    year="2015",
    publisher="Springer International Publishing",
    address="Cham",
    pages="17--55",
    doi={10.1007/978-3-319-09949-1_2}
}

@inbook{Rejzner2022, place={Cambridge}, title={BV Quantisation in Perturbative Algebraic QFT: Fundamental Concepts and Perspectives}, booktitle={The Philosophy and Physics of Noether’s Theorems: A Centenary Volume}, publisher={Cambridge University Press}, author={Rejzner, Kasia}, editor={Read, James and Teh, Nicholas J.}, year={2022}, pages={100–-133}}

@article{Brouder2014,
doi = {10.1088/1751-8113/47/44/443001},
year = {2014},
volume = {47},
number = {44},
author = {Brouder, Christian and Dang, Nguyen Viet and Hélein, Frédéric},
title = {A smooth introduction to the wavefront set},
journal = {Journal of Physics A: Mathematical and Theoretical},
}

@book{Kolar+etal,
    author  =   {Ivan Kol\'a\v{r} and Peter W. Michor and Jan Slov\'ak},
    title   =   {Natural Operations in Differential Geometry},
    publisher   =   {Springer-Verlag},
    address =   {New York},
    year    =   {1993},
}

@article{Marchtworoutes,
    author = {Eleanor March},
    title = {Two routes to the Hadamard condition},
    year = {2026}
}

@article{Brunettietal2019,
    author = {Brunetti, Romeo and Fredenhagen, Klaus and Ribeiro, Pedro Lauridsen},
    title = {Algebraic Structure of Classical Field Theory: Kinematics and Linearized Dynamics for Real Scalar Fields},
    journal = {Communications in Mathematical Physics},
    year = {2019},
    pages = {519--584},
    volume = {368},
    number = {2},
    doi = {10.1007/s00220-019-03454-z}
}

@misc{LinnemannRead,
            year = {2026},
           title = {Prising apart the Great "Puzzle" Ball of Thread: another look at the local validity of special relativity},
          author = {Niels Linnemann and James Read},
             url = {https://philsci-archive.pitt.edu/28714/}
}

@article{Baer2015,
    author = {Christian B{\"a}r},
    title = {Green-Hyperbolic Operators on Globally Hyperbolic Spacetimes},
    journal = {Communications in Mathematical Physics},
    year = {2015},
    pages = {1586--1615},
    volume = {333},
    number = {3},
    doi = {10.1007/s00220-014-2097-7}
}

@article{Khavkine2014,
    author = {Khavkine, Igor},
    title = {Covariant phase space, constraints, gauge and the Peierls formula},
    journal = {International Journal of Modern Physics A},
    volume = {29},
    number = {5},
    pages = {1430009},
    year = {2014},
    doi = {10.1142/S0217751X14300099}
}

@article{Radzikowski,
    author = {Radzikowski, Marek J.},
    title = {Micro-local approach to the Hadamard condition in quantum field theory on curved space-time},
    journal = {Communications in Mathematical Physics},
    year = {1996},
    pages = {529--553},
    volume = {179},
    number = {3},
    doi = {10.1007/BF02100096}
}

@article{BrunettiFredenhagen2000,
    author = {Romeo Brunetti and Klaus Fredenhagen},
    title = {Microlocal Analysis and Interacting Quantum Field Theories: Renormalization on Physical Backgrounds},
    journal = {Communications in Mathematical Physics},
    year = {2000},
    pages = {623--661},
    volume = {208},
    number = {3},
    doi = {10.1007/s002200050004}
}

@article{Fewster2000,
    doi = {10.1088/0264-9381/17/9/302},
    year = {2000},
    volume = {17},
    number = {9},
    pages = {1897},
    author = {Christopher J. Fewster},
    title = {A general worldline quantum inequality},
    journal = {Classical and Quantum Gravity}
}

@article{FewsterVerch2002,
    author = {Fewster, Christopher J. and Verch, Rainer},
    title = {A Quantum Weak Energy Inequality for Dirac Fields in Curved Spacetime},
    journal = {Communications in Mathematical Physics},
    year = {2002},
    pages = {331--359},
    volume = {225},
    number = {2},
    doi = {10.1007/s002200100584}
}

@inbook{Fewster2017,
    author="Fewster, Christopher J.",
    editor="Lobo, Francisco S. N.",
    title="Quantum Energy Inequalities",
    booktitle="Wormholes, Warp Drives and Energy Conditions",
    year="2017",
    publisher="Springer International Publishing",
    place="Cham",
    pages="215--254",
    doi="10.1007/978-3-319-55182-1_10"
}

@article{HollandsWald2001,
    author = {Stefan Hollands and Robert M. Wald},
    title = {Local Wick Polynomials and Time Ordered Products of Quantum Fields in Curved Spacetime},
    journal = {Communications in Mathematical Physics},
    year = {2001},
    volume = {223},
    pages = {289--326},
    number = {2},
    doi = {10.1007/s002200100540}
}

@article{HollandsWald2002,
    author = {Stefan Hollands and Robert M. Wald},
    title = {Existence of Local Covariant Time Ordered Products of Quantum Fields in Curved Spacetime},
    journal = {Communications in Mathematical Physics},
    year = {2002},
    volume = {231},
    pages = {309--345},
    number = {2},
    doi = {10.1007/s00220-002-0719-y}
}

@inbook{Fletcher2020,
    author = {Fletcher, Samuel C.},
    editor = {Beisbart, Claus and Sauer, Tilman and W{\"u}thrich, Christian},
    title = {Approximate Local Poincar{\'e} Spacetime Symmetry in General Relativity},
    booktitle = {Thinking About Space and Time: 100 Years of Applying and Interpreting General Relativity},
    year = {2020},
    publisher = {Springer},
    address= {Cham},
    pages = {247-267},
    doi = {10.1007/978-3-030-47782-0_12}
}

@article{Weatheralldogmas,
    author = {James Owen Weatherall},
    title = {Two Dogmas of Dynamicism},
    journal = {Synthese},
    year = {2021},
    pages = {253-275},
    volume = {199},
    issue = {Suppl 2},
    doi = {10.1007/s11229-020-02880-0}
}

@article{Peierls1952,
    author = {Peierls, Rudolf Ernst},
    title = {The commutation laws of relativistic field theory},
    journal = {Proceedings of the Royal Society of London. A. Mathematical and Physical Sciences},
    volume = {214},
    number = {1117},
    pages = {143--157},
    year = {1952},
    doi = {10.1098/rspa.1952.0158},
}

@article{MarchWeatherallFun1,
    author = {Eleanor March and James Owen Weatherall},
    title = {Fun with Minimal Coupling, Part 1: SEP, Local Approximate Poincare Symmetry, and All That},
    year = {2026}
}

\end{document}